\documentclass[journal]{IEEEtran}

\IEEEoverridecommandlockouts

\usepackage{bm}
\usepackage{epsf}
\usepackage{subfigure}
\usepackage{cite}
\usepackage{graphics}
\usepackage{epsfig} 
\usepackage{amsthm}
\usepackage{mathptmx} 
\usepackage{times} 
\usepackage{amsmath} 
\usepackage{amssymb}  
\usepackage{psfrag}
\usepackage{enumerate}
\usepackage{url}
\usepackage{stfloats}
\usepackage{array}
\usepackage{latexsym}
\usepackage{amssymb} 
\usepackage[usenames]{color}
\usepackage[ruled,vlined]{algorithm2e}
\usepackage{wrapfig}
\usepackage{booktabs}
\usepackage{threeparttable}

\newtheorem{definition}{Definition}
\newtheorem{theorem}{Theorem}
\newtheorem{problem}{Problem}
\newtheorem{proposition}{Proposition}

\newtheorem{remark}{Remark}

\newcommand{\lfteqn}{\begin{eqnarray} \begin{array}{lllllll}}
\newcommand{\ndeqn}{\end{array} \nonumber \end{eqnarray}}
\newcommand{\Lfteqn}{\begin{eqnarray} \begin{array}{lllllll}}
\newcommand{\Ndeqn}{\end{array}  \end{eqnarray}}

\begin{document}

\title{\LARGE \bf Distributed Nonblocking Supervisory Control of Timed Discrete-Event Systems with Communication Delays and Losses
}

\author{Yunfeng~Hou and Qingdu~Li
\thanks{Yunfeng Hou (yunfenghou@usst.edu.cn) and Qingdu Li (liqd@usst.edu.cn)  are with the Institute of Machine Intelligence, University of Shanghai for Science and Technology, Shanghai, 200093, China.}}
%Ching-Yen Weng (weng0025@e.ntu.edu.sg) is with the Robotics Research
%Centre, Nanyang Technological University, Singapore.}}

\maketitle

\begin{abstract}
This paper investigates the problem of distributed  nonblocking supervisory control for timed discrete-event systems (DESs).
The distributed supervisors communicate with each other over  networks subject to nondeterministic communication  delays and losses.
Given that the delays are counted by time, techniques have been developed to model the dynamics of the communication channels.
By incorporating the dynamics of the communication channels into the system model, we construct a communication automaton to model the interaction process between the supervisors.
Based on the communication automaton, we define the observation mappings for the supervisors, which consider delays and losses occurring in the communication channels. 
Then, we derive the necessary and sufficient conditions for
the existence of a set of  supervisors for distributed nonblocking  supervisory control.
These conditions are expressed as network
controllability, network joint observability, and system language closure. 
Finally, an example of intelligent manufacturing is provided to show the application of  the proposed framework.
\end{abstract}


\section{Introduction}

The supervisory control theory  is the core theory of discrete-event systems (DESs) and was initiated in the early 1980s \cite{ramadge1987supervisory,lin88is2}.
%The supervisory controller in DESs disables events that lead to some undesirable event sequences
%By incorporating the partial observation into the DES models \cite{lin88is2}, it has been extensively investigated over the past decades.
%Roughly speaking, the supervisory control theory studies how to enable or disable events according to the observed strings for achieving a specification language.
%As the development of affordable high-performance processors, sensing devices and wireless ad hoc networks, distributed supervisory control of DESs is approved to be implemented in a large spectrum of applications, such as distributed robot system and intel,ligent transportation system.
%Distributed supervisory control involves data communication among the supervisors.
Due to the limitation of the controller's memory, it is
sometimes not possible to control a large-scale system using a
single supervisor.
In this regard, a decentralized supervisory control approach was proposed in \cite{Rudie92tac,lin88is} such that several supervisors work as a group to control the system but do not communicate with each other.
After that, as the development of wireless ad hoc networks, a distributed supervisory control approach that allows the supervisors to exchange information of event occurrences over networks was considered in \cite{Barrett00TAC,Rudie03TAC}, where communications between the supervisors always can be performed instantaneously.
%All aforementioned works assume that the communication between the distributed supervisors are instantaneous and reliable.
However, due to the network characteristics, undetermined delays and losses always exist in 
data communication between the local supervisors, especially when the communication distance is long.
Thus, maintaining safety of the system under communication delays and losses is currently an active subject of research for real-life  applications.

In \cite{Tripakis04TAC,Hiraishi09TAC},  by assuming that the delays are upper bounded by $k$ event occurrences, the  supervisory control problem involving communication delays among multiple supervisors was studied.
%It is shown that the distributed supervisory control problem is undecidable in the unbounded-delay case.
Under similar assumptions as in \cite{Tripakis04TAC,Hiraishi09TAC}, the distributed fault diagnosis problem has also been addressed in \cite{Qiu08TAC}.
However, as we have discussed in \cite{Yunfeng23IJC}, it is difficult to determine the upper bound of the communication delays using number of event occurrences, as the interarrival time between two successive event occurrences is usually nondeterministic.
In \cite{renyuan16ijc}, procedures are developed to check if there exists a set of delay-robust supervisors that can work properly under communication delays.
The work of \cite{renyuan16ijc} is further extended to timed DESs in \cite{renyuan22ijc}.
The frameworks of \cite{renyuan16ijc,renyuan22ijc}, however, require that an occurred controllable event cannot reoccur unless this event occurrence  has been communicated to the connected supervisors. 
This requirement slows down the operating speed of the system.
Reference \cite{yuting22icca} gives further insight into distributed supervisory control with delays by considering: (i) communication delays between the plant and the supervisors; (ii) communication delays between the supervisors.
Command execution automata are constructed in  \cite{yuting22icca} to ensure that a new control command can be executed only when an observable event has occurred in the plant.
This requirement needs additional communications from the sensors to the actuators, which can be costly.
More recently, given that all the events are observable, the authors in \cite{Moormann23TASE} have studied how to distribute a set of local supervisors from a monolithic global one by taking communication delays among these supervisors into consideration.
Reference \cite{Kalyon14tac} has investigated how to synthesize a set of distributed supervisors under communication delays such that the controlled system can never reach a bad state.
Some other related works can be found in \cite{shu14ac,lin20tcns,yao20csl,yunfeng23tcns}, where a set of supervisors are used to control or diagnose a DES with observation delays between the plant and the supervisors.
All the supervisors in \cite{shu14ac,lin20tcns,yao20csl,yunfeng23tcns} make decisions based on their own observations, i.e., there are no communications among them.

%, where the local controllers can only observe the behavior of
%their proper subsystem, and make control decisions based on their own state estimate to avoid bad states.

\begin{figure}
	\begin{center}
		\includegraphics[width=6.8cm]{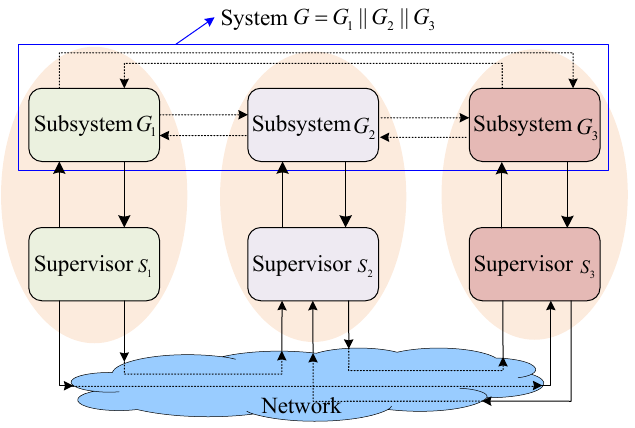}    % The printed column width is 8.4 cm.
		\caption{Distributed supervisory control over networks.} 
		\label{Fig1}
	\end{center}
\end{figure}

As shown in Fig. \ref{Fig1}, 
%we consider three independent  subsystems $G_1$, $G_2$, and $G_3$ over the respective event sets $\tilde{\Sigma}_1$, $\tilde{\Sigma}_2$, and $\tilde{\Sigma}_3$.
the system $G$ considered in this paper is composed of several subsystems  acting in parallel.
All the subsystems are modeled by
timed automata that share one common $tick$ event representing the elapse
of one unit of time.
The timed DESs were first proposed in \cite{brandin94tac} for supervisory control with full observation, and then is extended  to partial observation in \cite{lin95tac}.
%That is, there is a global clock.
Each subsystem is controlled by a supervisor.
These supervisors exchange information
about the event occurrences over networks subject to non-negligible delays and losses.
The communication delays are measured by the number of $tick$ occurrences.
To define the systems' observation mappings under possible delays and losses, we construct a communication automaton $\tilde{G}$ that incorporates the ``dynamics of the communication channels'' into  $G$.
After that, network controllability and network joint coobservability  are introduced
to capture whether there is a set of nonblocking  supervisors that can make sufficient observations under communication delays and losses.
Finally, an example is provided to show the application of the proposed framework.
%The necessary and sufficient conditions
%for the existence of a set of distributed supervisors under communication delays and losses are derived in terms of  network controllability and network joint coobservabilty.

The main contributions of this paper are  as follows.
Different from \cite{Tripakis04TAC,Hiraishi09TAC,Qiu08TAC}, this paper measures the delays using time.
In contrast to \cite{renyuan16ijc,renyuan22ijc,yuting22icca}, this paper does not impose  additional requirements on the event occurrence and the control command execution.
That is, an event that is active at the current state can occur as long as it is enabled by the current command, and a  command can be executed as long as it has been delivered to the actuator.
Compared with  \cite{Moormann23TASE,Kalyon14tac},  the system considered in this paper
can be partially observable, which is more realistic
because of the detection and communication limitations.
Different from \cite{shu14ac,lin20tcns,yao20csl,yunfeng23tcns}, this paper considers the communications between the supervisors.

%Due to the limitation of the number of the pages, the proofs of this paper are omitted and can be found in \cite{lafortune07book}.

The rest of this paper is organized  as follows.
Section II is the preliminaries.
Section III models how each supervisor interacts with the plant and the other supervisors. 
Section IV solves the distributed networked supervisory control problem.
Section V concludes this paper.

\section{Preliminaries}

\label{sec:ii}

%As depicted in Fig. \ref{Fig1}, we assume that the timed DESs consist of $n$ components.
In this paper, the system is modeled as a timed
automaton $G=(Q,\tilde{\Sigma},\delta,\Gamma,q_{0}, Q_{m}),$
where
$Q$ is the finite set of states;
$\tilde{\Sigma}=\Sigma\cup\{tick\}$ is the finite set of events with a special event $tick$ representing the elapse of one unit of time;
$\delta:Q \times \tilde{\Sigma} \rightarrow Q$ is the transition function;
$\Gamma: Q \rightarrow 2^{\tilde{\Sigma}}$ is the active event function;
$q_{0}$ is the initial state;
$Q_{m}$ is the set of marked states that are marked by double circles or blocks in this paper.
%$\tilde{\Sigma}^*$ is the Kleene Closure of $\tilde{\Sigma}$, i.e., the set of all strings over events
%in $\tilde{\Sigma}$.
$\delta$ is extended to the domain $Q \times {\tilde{\Sigma}}^*$ in the usual way \cite{lafortune07book}.
The languages generated and marked by $G$ are denoted by $\mathcal{L}(G)$ and $\mathcal{L}_m(G)$, respectively.
$\varepsilon$ is the empty string.
The prefix-closure of a string $s\in \mathcal{L}(G)$ is defined by $\overline{\{s\}}=\{u\in \tilde{\Sigma}^*:(\exists v\in \tilde{\Sigma}^*)uv=s\}$.
Given a language $L\in \tilde{\Sigma}^*$, $\overline{L}=\{u\in \tilde{\Sigma}^*:(\exists s\in L)u \in \overline{\{s\}}\}$.
We say that $L$ is prefix-closed if $L=\overline{L}$.
{We say that $L$ is $\mathcal{L}_m(G)$-closed if $L=\overline{L} \cap \mathcal{L}_m(G)$.}
$G$ is nonblocking if $\mathcal{L}(G)=\overline{\mathcal{L}_m(G)}$.
$Ac(G)$ is the accessible part of $G$.

%For any event $\sigma \in \tilde{\Sigma}\setminus \{tick\}$ that is active at the current state and is enabled, it will occur between a lower time bound $l_{\sigma}$ and an upper time bound $u_{\sigma}$.
%According to the upper time bound $u_{\sigma}$, an event in $\tilde{\Sigma}\setminus \{tick\}$ can be classed into a prospective event or a remote event.
%Specifically, we say that $\sigma$ is a prospective event if $u_{\sigma}$ is finite, and  a remote event if $u_{\sigma}$ is infinite.
%We denote by $\Sigma_{spe}$ the set of prospective events, and  $\Sigma_{rem}$ the set of  remote events.
%By definitions,  $\Sigma=\Sigma_{spe}\cup\Sigma_{rem}$ and $\tilde{\Sigma}=\Sigma_{spe}\cup\Sigma_{rem}\cup\{tick\}$.

Some events in ${\Sigma}$ can be enforced since they can pre-empt the occurrence of $tick$.
We denote $\Sigma_{for}\subseteq {\Sigma}$ by the set of enforceable events.
%In fact, $tick$ is between a controllable event and an uncontrollable event, since we cannot stop the time from happening, but we can pre-empt the occurrence of $tick$ by executing an enforceable event.
%We also have $\tilde{\Sigma}=\Sigma_{spe}\cup\Sigma_{rem}\cup\{tick\}$.
We assume that $G$ satisfies the following three conditions: 1) A finite number events can occur in one unit of time. i.e., $(\forall q \in Q)(\forall s \in \Sigma^*\setminus \{\varepsilon\})\delta(q,s)\neq q$; 2) Time will never stop, i.e., $(\forall q \in Q)(\exists \sigma \in \tilde{\Sigma})\delta(q,\sigma)!$; 3) If no $tick$ is active after a string, some enforceable events must be active after this string, i.e., $(\forall s \in \mathcal{L}(G))s\ tick \notin \mathcal{L}(G)\Rightarrow (\exists \sigma \in \Sigma_{for})s\sigma \in \mathcal{L}(G)$.

%Given two timed automata  $G_1=(Q_1,\tilde{\Sigma}_1,\delta_1,q_{0,1}, Q_{m,1})$ and $G_2=(Q_2,\tilde{\Sigma}_2,\delta_2,q_{0,2}, Q_{m,2})$, we define $G_1 \parallel G_2$ as the parallel composition of $G_1$ and $G_2$ \cite{lafortune07book}.
%In this paper, we suppose that the DES $G$ is a composed system and is defined by the parallel composition of a set of subplants $G_1||G_2||\cdots||G_n$.
%All the subplants work individually and share a global clock.
%For each subplant $G_i$, we denote by $\tilde{\Sigma}_i=\Sigma_i \cup\{tick\}$, where $\Sigma_i$ is the private event set of subplant $G_i$, $i\in \{1, \ldots, n\}$.
Given two timed automata $G_1$ and $G_2$, we say that $G_1 \parallel G_2$ is the parallel composition of $G_1$ and $G_2$ \cite{lafortune07book}.
We say that $G_1$ is a subautomaton of $G_2$, denoted by $G_1 \sqsubseteq G_2$, if $G_1$ can be obtained from $G_2$ by removing some states of $G_2$ and all the transitions connected to these states.
The system $G$ is defined as the parallel composition of $n$  subsystems $G=G_1||G_2||\cdots||G_n$, where $G_i=(Q_i,\tilde{\Sigma}_i,\delta_i,\Gamma_i,q_{0,i}, Q_{m,i})$ for $i=1,2,\ldots,n$.
All the subsystems are independent but share a globle clock.
Thus, all the events in $G_i$ except for $tick$ are disjoint, i.e., for all $i,j \in \mathcal{A}$ with $i \neq j$,  $\tilde{\Sigma}_{i} \cap \tilde{\Sigma}_{j}=\{tick\}$ and $(\tilde{\Sigma}_{i}\setminus \{tick\}) \cap (\tilde{\Sigma}_{j} \setminus \{tick\})=\emptyset$.

Let $K \subseteq \mathcal{L}(G)$ be the specification language given as the control objective.
%Given two automata $G_1$ and $G_2$, we say that $G_1$ is a subautomaton of $G_2$, denoted by $G_1 \sqsubseteq G_2$, if $G_1$ can be obtained from $G_2$ by removing some states of $G_2$ and all the transitions connected to these states.
%We suppose that $K$ can be represented by a subautomaton 
%$H \sqsubseteq G$ of $G$.
In many applications, the original $G$ may not satisfy
 $K$. 
To make the system fulfill $K$,  a set of distributed supervisors is used to control $G$.
We let $\mathcal{A}=\{1,2,\ldots,n\}$ be the index set of the distributed supervisors.
Each supervisor $i\in\mathcal{A}$ is responsible for a subsystem $G_i$.
%the system $G$ consists of $n$ distributed components, denoted by components $1, \ldots, n$.
%We assume that (i) all the components share one and only one a common event $tick$ (there is a global clock), and all their other events are disjoint; (ii) a supervisor is  responsible for one and only one component, and vice versa.
%We denote the set of events that are associated with component $i$ and supervisor $i$ as $\tilde{\Sigma}_i \subseteq \tilde{\Sigma}$.
For each local supervisor $i\in \mathcal{A}$, we denote $\tilde{\Sigma}_{c,i} \subseteq \tilde{\Sigma}_i$ and $\tilde{\Sigma}_{uc,i} = \tilde{\Sigma}_i \setminus \tilde{\Sigma}_{c,i}$ by the set of its controllable and uncontrollable events, respectively.
We denote $\tilde{\Sigma}_{o,i} \subseteq \tilde{\Sigma}_i$  and $\tilde{\Sigma}_{uo,i} = \tilde{\Sigma}_i \setminus \tilde{\Sigma}_{o,i}$ by the set of its observable and unobservable events, respectively.
We let $\tilde{\Sigma}_{uc}=\tilde{\Sigma}_{uc,1}\cup\cdots\cup \tilde{\Sigma}_{uc,n}$, $\tilde{\Sigma}_c=\tilde{\Sigma}\setminus \tilde{\Sigma}_{uc}$ and $\tilde{\Sigma}_{uo}=\tilde{\Sigma}_{uo,1}\cup\cdots\cup \tilde{\Sigma}_{uo,n}$, $\tilde{\Sigma}_o=\tilde{\Sigma}\setminus \tilde{\Sigma}_{uo}$.
%We write $\tilde{\Sigma}_{uc}=\tilde{\Sigma} \setminus \tilde{\Sigma}_c$ and $\tilde{\Sigma}_{uo}=\tilde{\Sigma} \setminus \tilde{\Sigma}_o$.
For any $\sigma \in \tilde{\Sigma}$, we denote $\mathcal{A}^c(\sigma)=\{i \in \mathcal{A}:\sigma \in \tilde{\Sigma}_{c,i}\}$ by the set of supervisors who can control the occurrence of $\sigma$.

%Additionally, we denote $\Sigma_{uc,i}=\tilde{\Sigma}\setminus \tilde{\Sigma}_{c,i}$ and $\Sigma_{uo,i}=\tilde{\Sigma} \setminus \tilde{\Sigma}_{o,i}$ by the set of uncontrollable events and unobservable events of supervisor $i$, respectively.

%In the context of distributed networked supervisory control, the supervisors communicate with each other over networks.
We use a boolean matrix $\textbf{COM} \in \{0,1\}^{n \times n}$ to describe the communication topology between the $n$ supervisors.
We denote $\textbf{COM}_{ij}$ by the boolean value in row $i$ and column $j$  of $\textbf{COM}$.
For supervisors $i,j \in \mathcal{A}$, there is a communication from supervisors $i$ to  $j$ iff $\textbf{COM}_{ij}=1$.
That is, $\textbf{COM}_{ij}=1$ indicates that supervisor $j$ can communicate information to supervisor $i$.
Meanwhile, if there is no communication from supervisors $i$ to $j$, we have $\textbf{COM}_{ij}=0$.
Note that $\textbf{COM}_{ii}=0$ for all $i\in \mathcal{A}$, i.e., a supervisor does not need to communicate with itself.
$tick$ can be sensed by all the supervisors without any delays and losses.
For supervisors $i,j\in \mathcal{A}$, if $\textbf{COM}_{ij}=1$,  we let $\Sigma_{ij} \subseteq \tilde{\Sigma}_{o,i}\setminus \{tick\}$ be the set of events that supervisor $i$ communicates to supervisor $j$.
That is, when an event $\sigma \in \Sigma_{ij}$ occurs in plant $G$, it will be communicated from supervisors $i$ to  $j$.
%Note that if $\textbf{COM}_{ij}=0$, $\Delta_{ij}=\emptyset$.
 %  $tick \notin \Sigma_{ij}$ for all $i, j\in \mathcal{A}$.
We denote $\Sigma_{L,ij} \subseteq \Sigma_{ij}$ by the set of events that may be lost  when supervisor $i$ communicates with supervisor $j$.

The communications between the supervisors
are carried out over a shared network, which incur communication delays and losses.
We denote the communication channel from supervisor $i$ to supervisor $j$ by $\textbf{CH}_{ij}$.
We make the following assumptions on $\textbf{CH}_{ij}$: 1) First-in-first-out (FIFO) is satisfied, i.e., the events queued at $\textbf{CH}_{ij}$ are communicated to supervisor $j$ in the same order as they were observed by supervisor $i$; 2) the communication delays are
upper bounded by $N_{ij}$ $tick$ occurrences, i.e., any event delayed at $\textbf{CH}_{ij}$ can be communicated (if no loss) before no more than $N_{ij}$ units of time; 3) The communication losses occurring in $\textbf{CH}_{ij}$ are nondeterministic, i.e., for any event delayed at $\textbf{CH}_{ij}$, if it is defined in $\Sigma_{L,ij}$, then it can be lost at any time during the communication.

\section{Modeling the distributed control systems}

\subsection{Modeling the communication channels}

\begin{definition}\label{DefII}
Given two supervisors $i,j \in \mathcal{A}$ with $\mathbf{COM}_{ij}=1$, the communication channel $\mathbf{CH}_{ij}$ configuration is defined as $\theta_{ij}=(\sigma_1,n_1)\cdots(\sigma_k,n_k),$
where $\sigma_1 \cdots \sigma_k \in \Sigma_{ij}^*$ is a sequence of communication events in $\Sigma_{ij}$ being transmitted from supervisor $i$ to supervisor $j$, and $n_d \in [0,N_{ij}]$, $d=1,\ldots,k$ is the number of  $tick$ occurrences since $\sigma_d$ has been pushed into $\mathbf{CH}_{ij}$.
\end{definition}
%\begin{definition}\label{DefIII}
%For any supervisor $i \in \mathcal{A}$, the {control channel ${CC}_{i}$ configuration} from supervisor $i$ to the plant is defined as a sequence of pairs: $\theta_{c,i}=(\sigma_1,n_1)\ldots(\sigma_k,n_k),$
%where $\sigma_1 \cdots\sigma_k \in \Sigma_{o}^*$ is a sequence of observable or communication events that have been delivered to supervisor $i$ but the control commands made following these communications are currently delayed at the control channel ${CC}_{i}$, and $n_z \in [0,N_{c,i}]$, $z=1,\ldots,k$ is the number of  $tick$ occurrences since the communication of $\sigma_z$.
%\end{definition}
%If the observation channel is empty, $\theta_o=[\varepsilon,0]$.
We denote $\Theta_{ij} \subseteq ({\Sigma_{ij} \times [0,N_{ij}]})^{\le T_{ij}} $ by the set of all the communication channel $\textbf{CH}_{ij}$ configurations, where $T_{ij} \in \mathbb{N}$ is the maximum length of events delayed at  $\textbf{CH}_{ij}$.
Since the delays are upper bounded
by $N_{ij}$, %all the communication events in $\Sigma_{ij}$ that occur in the past $N_{ij}$ units of times may be delayed at the communication channel $\textbf{CH}_{ij}$. 
the
number of events delayed at $\textbf{CH}_{ij}$ is upper bounded by the maximum number of 
events in $\Sigma_{ij}$ that may occur in $N_{ij}$ units of times.
Since only a finite number of events can occur in one $tick$, $T_{ij}$ is finite.

By adding $\theta_{ij}=\varepsilon$ for all $i,j \in \mathcal{A}$ such that $\textbf{COM}_{ij}=0$,  the state of the communication channels is defined as $\bar{\theta}=[\theta_{11},\ldots,\theta_{1n},\ldots,\theta_{n1},\ldots,\theta_{nn}],$ where $\theta_{ij}\in \Theta_{ij}$ is the communication channel $\textbf{CH}_{ij}$ configuration.
Let $\Theta=\Theta_{11} \times \cdots \times \Theta_{1n} \times \cdots \times \Theta_{n1} \times \cdots \times \Theta_{nn}$ be the set of all communication channel configurations.
We need the following notions to proceed..
%We note that $\bar{\theta} \in \Theta$ is not a constant but can be changed.
%Before we show how to update $\bar{\theta}$, we first introduce the following notations.
%We next introduce the following operators to update each $\bar{\theta}\in \Theta$.
\begin{itemize}
\item 
Given any $\theta_{ij} \in \Theta_{ij}$, if $\theta_{ij}=(\sigma_1,n_1)\cdots(\sigma_k,n_k) \neq \varepsilon$, let $\textbf{MAX}(\theta_{ij})=n_1$ be the maximum delays occurring in $\textbf{CH}_{ij}$, and if $x_{ij} = \varepsilon$, let $\textbf{MAX}(\theta_{ij})=0$;
\item
To update $\bar{\theta}$ after a $tick$, we define  ``$+$'' on $\theta_{ij}$ as follows.
For any $\theta_{ij} \in \Theta_{ij}$, (i) if $\theta_{ij}=\varepsilon$, $\theta_{ij}^+=\varepsilon$; and (ii) if $\theta_{ij}=(\sigma_1,n_1)\cdots(\sigma_k,n_k)\neq \varepsilon$, $\theta_{ij}^+=(\sigma_1,n_1+1)\cdots(\sigma_k,n_k+1)$;
\item
For any $\bar{\theta}=[\theta_{11},\ldots,\theta_{1n},\ldots,\theta_{n1},\ldots,\theta_{nn}]\in \Theta$ and any $i,j \in \mathcal{A}$,
$
\textbf{REP}_{ij}(\bar{\theta},\theta)
$
 replaces $\theta_{ij}$ by $\theta$ in $\bar{\theta}$, without changing the remaining elements of $\bar{\theta}$.
\end{itemize}
%By definitions, when a $tick$ occurs, all the natural numbers in $\theta_{ij}$ plus 1 for recording the communication delays.
The following operators are defined to update  $\bar{\theta}\in \Theta$.
%Next, we show how to update each ${\theta}_{ij}\in \Theta_{ij}$.

%\footnote{By assumption, before an observable event is communicated, there are at most $N_o$ additional event occurrences. Therefore, the length of $\theta_o\in \Theta_{o}$ is no longer than $N_o+1$.}.  let $[\theta_o]_1$ and $[\theta_o]_2$ be the respective first and second components of $\theta_o$, i.e., $[\theta_o]_1=w \wedge [\theta_o]_2=n$; \footnote{If $w = \varepsilon$, since there are no observable events queued at the observation channel, the delay occurring in the observation channel is still 0  upon the occurrence of $tick$. Hence,  we set $\textbf{MAX}(\varepsilon)=-1$,  or equivalently,  $\textbf{MAX}(\varepsilon)+1=0$, if $[\theta_o]_1 = \varepsilon$.

\begin{enumerate}
\item 
When a $tick$ occurs in $G$, we define $\textbf{TIME}:\Theta  \rightarrow \Theta$ as follows: for all $\bar{\theta}=[\theta_{11},\ldots,\theta_{1n},\ldots,\theta_{n1},\ldots,\theta_{nn}] \in \Theta$,
 \Lfteqn\label{Eq1}
\textbf{TIME}(\bar{\theta})= \begin{cases}
   \bar{\theta}'  &\text{if}\ (\forall i,j\in \mathcal{A})\textbf{MAX}(\theta_{ij}^+) \le N_{ij}\\ 
    \not ! & \text{otherwise}
  \end{cases}
  \Ndeqn  
where $\bar{\theta}'=[\theta_{11}^+,\ldots,\theta_{1n}^+,\ldots,\theta_{n1}^+,\ldots,\theta_{nn}^+]$.

\item When a $\sigma \in \tilde{\Sigma}\setminus\{tick\}$ occurs in $G$,  define $\textbf{IN}: {\Theta} \times \Sigma \rightarrow \Theta$ as follows: for all  $\bar{\theta}=[\theta_{11},\ldots,\theta_{1n},\ldots,\theta_{n1},\ldots,\theta_{nn}] \in \Theta$ and all $\sigma \in {\Sigma}$, 
  \Lfteqn\label{Eq2}
  \textbf{IN}(\bar{\theta},\sigma)=[\theta'_{11},\ldots,\theta'_{1n},\ldots,\theta'_{n1},\ldots,\theta'_{nn}]
\Ndeqn  
where if $\sigma \in \Sigma_{ij}$, then $\theta_{ij}'=\theta_{ij}(\sigma,0)$, and if $\sigma \notin \Sigma_{ij}$, then $\theta_{ij}'=\theta_{ij}$ for all $i,j\in \mathcal{A}$.

\item When a $\sigma \in {\Sigma}_{ij}$ delayed at $\textbf{CH}_{ij}$ is communicated, we define $\textbf{OUT}_{ij}:\Theta \times {\Sigma}_{ij}  \rightarrow  \Theta$ as follows: for all  $\bar{\theta}=[\theta_{11},\ldots,\theta_{1n},\ldots,\theta_{n1},\ldots,\theta_{nn}] \in \Theta$ and all $\sigma \in {\Sigma}_{ij}$,
 \Lfteqn\label{Eq3}
  \textbf{OUT}_{ij}(\bar{\theta},\sigma)= \begin{cases}
   \bar{\theta}' &\text{if}\ \theta_{ij}=(\sigma_1,n_1)\cdots (\sigma_k,n_k)\\
& \neq \varepsilon \wedge \sigma_1=\sigma\\ 
    \not ! & \text{otherwise},
  \end{cases}
  \Ndeqn
where $\bar{\theta}'=\textbf{REP}_{ij}(\bar{\theta},\theta_{ij} \setminus (\sigma_1,n_1))$.

\item When the $d$th event is lost from the communication channel $\textbf{CH}_{ij}$, we define $\textbf{LOSS}_{ij}:\Theta \times \mathbb{N}  \rightarrow  \Theta$ as follows: for all $\bar{\theta}=[\theta_{11},\ldots,\theta_{1n},\ldots,\theta_{n1},\ldots,\theta_{nn}] \in \Theta$ and all $d \in \mathbb{N}$,
 \Lfteqn\label{Eq4}
  \textbf{LOSS}_{ij}(\bar{\theta},d)= \begin{cases}
   \bar{\theta}'  &\text{if}\ \theta_{ij}=(\sigma_1,n_1)\cdots (\sigma_k,n_k)\\ 
   & \neq \varepsilon \wedge d \le k \wedge \sigma_d \in \Sigma_{L,ij}\\
    \not ! & \text{otherwise},
  \end{cases}
  \Ndeqn
where $\bar{\theta}'=\textbf{REP}_{ij}(\bar{\theta},\theta)$ with $$\theta=(\sigma_1,n_1)\cdots (\sigma_{d-1},n_{d-1})(\sigma_{d+1},n_{d+1}) \cdots (\sigma_k,n_k).$$

\end{enumerate}

%\footnote{Since $tick$ can be observed without any delays and losses, the occurrence of $tick$ can be sensed at the time it occurs. Therefore, we do not insert $tick$ into $\theta_o$ in (\ref{Eq1}).}

Equation (\ref{Eq1}): Since the delays occurring in $\textbf{CH}_{ij}$ are no larger than $N_{ij}$, $\textbf{TIME}(\bar{\theta})$ is defined iff $\textbf{MAX}(\theta_{ij}^+) \le N_{ij}$ for all $i,j \in \mathcal{A}$.
When a $tick$ occurs, by the definition of ``$+$'', all the $\theta_{ij}$ should be updated to $\theta_{ij}^+$.
Equation (\ref{Eq2}): When an event $\sigma \in \Sigma$ occurs in $G$, if  $\sigma \in \Sigma_{ij}$, it will be pushed into the channel $\textbf{CH}_{ij}$.
By FIFO, $\textbf{IN}_{ij}(\bar{\theta},\sigma)$ adds $(\sigma,0)$ to the end of  $\theta_{ij}$.
For the remaining $\theta_{rt}$ such that $r \neq i \lor t \neq j$, we keep them unchanged as $\sigma \notin \Sigma_{rt}$.
Equation (\ref{Eq3}): When an event queued at $\textbf{CH}_{ij}$ is communicated, by FIFO, it must be the first event.
Thus,  $\textbf{OUT}_{ij}(\bar{\theta},\sigma)$ is defined if $\theta_{ij}=(\sigma_1,n_1)\ldots(\sigma_k,n_k)\neq \varepsilon \wedge \sigma_1=\sigma$.
When $\sigma$ is communicated,  $\textbf{OUT}_{ij}(\bar{\theta},\sigma)$ removes  $(\sigma_1,n_1)$ from the head of $\theta_{ij}$ in $\bar{\theta}$.
%For the remaining $\theta_{rt}$ such that $r \neq i \lor t \neq j$ in $\bar{\theta}$, we keep them unchanged.
Equation (\ref{Eq4}): The $d$th event can be lost from $\textbf{CH}_{ij}$ if,  (i) $\textbf{CH}_{ij}$ is not empty, i.e., $\theta_{ij}=(\sigma_1,n_1) \cdots (\sigma_k,n_k)\neq \varepsilon$; (ii) the queue length of events delayed at $\textbf{CH}_{ij}$ is no smaller than $d$, i.e., $d \le k$; and (iii) the $d$th event queued at $\textbf{CH}_{ij}$ can be lost, i.e., $\sigma_d \in \Sigma_{L,ij}$.
Thus, $\textbf{LOSS}_{ij}(\bar{\theta},d)$ is defined, if $\theta_{ij}=(\sigma_1,n_1)\cdots (\sigma_k,n_k) \neq \varepsilon \wedge d \le k \wedge \sigma_d \in \Sigma_{L,ij}$.
If the $d$th  event is lost from $\textbf{CH}_{ij}$, $\textbf{LOSS}_{ij}(\bar{\theta},d)$ removes the $d$th component from $\theta_{ij}$  in $\bar{\theta}$.
%The remaining $\theta_{rt}$ such that $r \neq i \lor t \neq j$ are kept unchanged.

\subsection{Communication automaton}

Next, let us construct the communication automaton $\tilde{G}$ by incorporating ``the dynamics of the communication channels'' into states of system $G$.
Based on $\tilde{G}$, we define the observation mappings of  the supervisors under the communication delays and losses.
Let us first define two special types of event.
%Formally, we have the following definitions.
\begin{enumerate}
\item To keep track of what has been communicated from  $\textbf{CH}_{ij}$, we define bijection $f_{ij}:\Sigma_{ij} \rightarrow \Sigma^f_{{ij}}$ such that $\Sigma^f_{ij}=\{f_{ij}(\sigma):\sigma\in \Sigma_{ij}\}$.
For all $\sigma \in \Sigma_{ij}$, we use $f_{ij}(\sigma)$ to denote that the first event $\sigma \in \Sigma_{ij}$ delayed at $\textbf{CH}_{ij}$ is communicated.
We denote $\Sigma^f=\Sigma^f_{11} \cup \cdots \cup \Sigma^f_{1n} \cup  \cdots \cup \Sigma^f_{n1} \cup \cdots \cup \Sigma^f_{nn}$;

\item 
To keep track of what has been lost from $\textbf{CH}_{ij}$, we define bijection $g_{ij}:\mathbb{N} \rightarrow \Sigma^g_{{ij}}$ such that $\Sigma^g_{ij}=\{g_{ij}(d):d \in \mathbb{N}\}$.
For all $d \in \mathbb{N}$, we use $g_{ij}(d)$ to denote that the $d$th event delayed at $\textbf{CH}_{ij}$ is lost.
We denote $\Sigma^g=\Sigma^g_{11} \cup \cdots \cup \Sigma^g_{1n} \cup  \cdots \cup \Sigma^g_{n1} \cup \cdots \cup \Sigma^g_{nn}$.

\end{enumerate}

\begin{figure}
	\begin{center}
		\includegraphics[width=8.9cm]{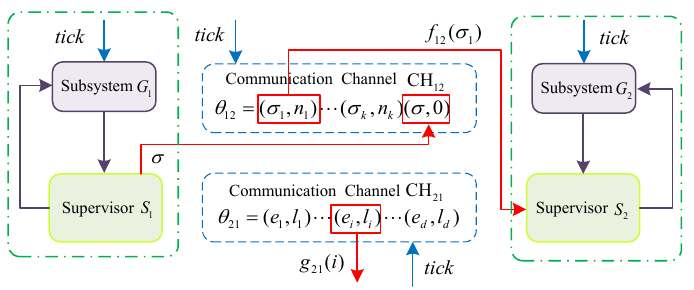}    % The printed column width is 8.4 cm.
		\caption{The distributed networked supervisory control system.} 
		\label{Fig2}
	\end{center}
\end{figure}

We  use Fig. \ref{Fig2} to gain some intuitions for constructing $\tilde{G}$.
Let $n=2$ for brevity.
%BWe suppose in Fig. \ref{Fig2} that there are two supervisors who can communicate with each other.
The state of $\tilde{G}$ is defined as a pair $\tilde{q}=[q,\bar{\theta}=(\theta_{12},\theta_{21})]\in Q \times \Theta$.\footnote{Here, we omit $\theta_{11}$ and $\theta_{22}$ of $\bar{\theta}$ as $\theta_{11}=\theta_{22}=\varepsilon$ all the time.}
When an event $\sigma\in \Sigma_{12}$ occurs in $G$, by definition, $(\sigma,0)$ will be added to the end of $\theta_{12}$.
By FIFO, events delayed at $\textbf{CH}_{12}$ are communicated in the same order as they occur.
As shown in Fig. \ref{Fig2}, $\sigma_1$  is the first event communicated to supervisor 2, followed by $\sigma_2$, $\sigma_3$, and so on.
Meanwhile, if $e_i \in \Sigma_{L,21}$, then $e_i$ may be lost from $\textbf{CH}_{21}$.
And if $e_i$ is lost, the $i$th element of $\theta_{21}$, i.e., $(e_i,l_i)$ will be removed from $\theta_{21}$.
Formally, we construct $\tilde{G}=Ac(\tilde{Q}, \tilde{E}, \tilde{\delta}, \tilde{\Gamma},\tilde{q}_{0},\tilde{Q}_{m}),$  where  $\tilde{Q} \subseteq Q \times \Theta$ is the state space;
$\tilde{E} \subseteq \tilde{\Sigma} \cup \Sigma^f \cup \Sigma^g$ is the event set;
$\tilde{\Gamma}:\tilde{Q}\rightarrow 2^{\tilde{E}}$ is the active event function;
$\tilde{q}_{0}=(q_{0},\underbrace{\varepsilon, \ldots,\varepsilon}_{n \times n})$ is the initial state; 
$\tilde{Q}_m=\{(q,\bar{\theta}) \in \tilde{Q}: q \in Q_m \}$ is the set of marked states;
the transition function $\tilde{\delta}: \tilde{Q} \times \tilde{E} \rightarrow \tilde{Q}$ is defined as follows:

\begin{itemize}
%  \item  For $\tilde{q}_0 \in \tilde{Q}$ and $g(\varepsilon)=\Sigma_g$, let $\tilde{\delta}(\tilde{q}_0, g(\varepsilon))=(q_0,\varepsilon,\varepsilon)$;

 \item For any $\tilde{q}=({q},\bar{\theta}) \in \tilde{Q}$ and event $tick$, $\tilde{\delta}(\tilde{q}, tick)=$
 \Lfteqn\label{Eq5}
  \begin{cases}
  [\delta({q},tick),\mathbf{TIME}(\bar{\theta})]  & \text{if}\ \delta({q},tick)! \\
&\wedge \mathbf{TIME}(\bar{\theta})! \\
   \not ! & \mbox{otherwise.}
  \end{cases}
  \Ndeqn 

 \item For any $\tilde{q}=({q},\bar{\theta}) \in \tilde{Q}$ and $\sigma \in {\Sigma}$, $\tilde{\delta}(\tilde{q}, \sigma)=$
 \Lfteqn\label{Eq6}
  \begin{cases}
  [\delta({q},\sigma),\mathbf{IN}(\bar{\theta},\sigma)]  & \text{if}\ \delta({q},\sigma)!\\
  \not ! & \mbox{otherwise.}
  \end{cases}
  \Ndeqn

\item For any $\tilde{q}=({q},\bar{\theta}) \in \tilde{Q}$ and $f_{ij}(\sigma) \in \Sigma^f$, $\tilde{\delta}(\tilde{q}, f_{ij}(\sigma))= $

  \Lfteqn\label{Eq7}
  \begin{cases}
   [{q},\mathbf{OUT}_{ij}(\bar{\theta},\sigma)]  &\text{if}\ \mathbf{OUT}_{ij}(\bar{\theta}, \sigma)!\\ 
   \not ! & \text{otherwise.}
  \end{cases}
  \Ndeqn

\item For any $\tilde{q}=({q},\bar{\theta}) \in \tilde{Q}$ and $g_{ij}(d) \in \Sigma^g$, $ \tilde{\delta}(\tilde{q}, g_{ij}(d))= $

  \Lfteqn\label{Eq8}
 \begin{cases}
   [{q},\mathbf{LOSS}_{ij}(\bar{\theta},d)]  &\text{if}\ \textbf{LOSS}_{ij}(\bar{\theta},d)!\\ 
   \not ! & \text{otherwise.}
  \end{cases}
  \Ndeqn  

\end{itemize}

Equation (\ref{Eq5}): for any $\tilde{q}=({q},\bar{\theta}) \in \tilde{Q}$, $tick$ is defined at $\tilde{q}$ iff $tick$ is active at $\bar{q}$ and the delays occurring in  $\textbf{CH}_{ij}$ are no larger than $N_{ij}$.
Thus, $\tilde{\delta}(\tilde{q}, tick)$ is defined iff $\delta({q},tick)! \wedge \mathbf{TIME}(\bar{\theta})!$.
If $tick$ occurs in system $G$,  we update ${q}$ to $\delta({q},tick)$ for tracking the plant state, and update $\bar{\theta}$ to $\mathbf{TIME}(\bar{\theta})$ for updating the communication delays.
Equation (\ref{Eq6}): for any $\tilde{q}=({q},\bar{\theta}) \in \tilde{Q}$ and any $\sigma \in {\Sigma}$,  $\sigma$ is defined at $\tilde{q}$ iff $\delta({q},\sigma)!$.
When $\sigma$ occurs, we set ${q} \leftarrow \delta({q},\sigma)$ for tracking the plant state.
Meanwhile, by (\ref{Eq2}),  $\bar{\theta}\leftarrow \mathbf{IN}(\bar{\theta},\sigma)$.
%Otherwise, if  $\sigma \notin \Sigma_{ij}$, we keep $\bar{\theta}$ unchanged.
Equation (\ref{Eq7}): for any $\tilde{q}=({q},\bar{\theta}) \in \tilde{Q}$ and any $f_{ij}(\sigma) \in {\Sigma}^f$, by FIFO, $\sigma$ can be communicated from supervisors $i$ to $j$ iff $\sigma$ is the first event in $\textbf{CH}_{ij}$, i.e., $\mathbf{OUT}_{ij}(\bar{\theta},\sigma)!$.
%By (\ref{Eq2}), we know $\textbf{OUT}^{obs}(\theta_o,\sigma)!$.
When $\sigma$ is communicated, by (\ref{Eq3}), $\bar{\theta} \leftarrow \textbf{OUT}_{ij}(\bar{\theta},\sigma)$.
%We keep $\bar{q}$ unchanged since no events occur in $G$.
Equation (\ref{Eq8}): for any $\tilde{q}=({q},\bar{\theta}) \in \tilde{Q}$ and any $g_{ij}(d) \in {\Sigma}^g$, the $d$th event may be lost from $\textbf{CH}_{ij}$ iff $\textbf{LOSS}_{ij}(\bar{\theta},d)!$.
If the $d$th  event is lost from $\textbf{CH}_{ij}$, by (\ref{Eq4}), $\bar{\theta} \leftarrow \mathbf{LOSS}_{ij}(\bar{\theta},d)$.
%We keep $\bar{q}$ unchanged since no events occur in $G$.

%\begin{remark}
%It is worthy mentioning that when we construct $\tilde{G}$, we ``ignore'' the controls imposed on $G$, and we care only about communication delays and losses that may occur during the system evolution process.
%By  $\tilde{G}$, one can determine that (i) the string that has occurred in $G$, and (ii) the observable event occurrences that have been received by the supervisors.
%As will be discussed later, by $\tilde{G}$, we can ``decode'' the observations of the supervisors in the presence of nondeterministic communication delays and losses, when a string occurs in  $G$.
%\end{remark}

\begin{remark}

The computational complexity for the construction of $\tilde{G}$ is determined by the state space of $\tilde{G}$. 
By the definition of $\tilde{G}$, $|\tilde{Q}|$ is upper bounded by $|Q| \times |\Theta|$.
Since $\Theta=\Theta_{11} \times \cdots \times \Theta_{1n}\times \cdots \times \Theta_{n1}\times \cdots \times \Theta_{nn}$ and $\Theta_{ij}\subseteq (\Sigma_{ij}\times [0,N_{ij}])^{T_{ij}}$, $|\tilde{Q}|$ is upper bounded by
$\sum_{i=1}^n\sum_{j=1}^n |Q| \times |\Sigma|^{T_{ij}} \times (N_{ij}+1)^{T_{ij}}$.
Therefore, the complexity for constructing $\tilde{G}$ is polynomial with respect to $|Q|$, $|\Sigma|$, $n$, and $N_{ij}$ but is exponential with respect to $T_{ij}$.
\end{remark}

An example for the construction of $\tilde{G}$ will be provided in Section IV.
Given a  $\mu \in \mathcal{L}(\tilde{G})$, let $\psi(\mu)$ be the strings obtained by removing all the events in $\Sigma^f\cup\Sigma^g$ from $\mu$, without changing the order of the remaining events.
%Intuitively, $\psi(\mu)$ tracks the string that has occurred in $G$.
Given a  $L \subseteq \mathcal{L}(\tilde{G})$, let $\psi(L)=\{\psi(\mu):\mu\in L\}$.
With the above preparisions, we define the observation mapping $\psi^{f_i}(\cdot):\tilde{\Sigma}^* \rightarrow \tilde{\Sigma}^*_o$  for supervisor $i$ as follows:  $\psi^{f_i}(\varepsilon)=\varepsilon$, and for all $\mu,\mu e\in \tilde{\Sigma}^*$,
 \Lfteqn\label{Eq10}
  \psi^{f_i}(\mu e)= \begin{cases}
   \psi^{f_i}(\mu) e  &\text{if}\ e \in \tilde{\Sigma}_{o,i}\\ 
   \psi^{f_i}(\mu) \sigma  &\text{if}\ e=f_{ji}(\sigma) \in \Sigma^f\\ 
    \psi^{f_i}(\mu) & \text{otherwise}.
  \end{cases}
  \Ndeqn
Intuitively, for all $\mu \in \mathcal{L}(G_S)$, $\psi(\mu)$ tracks the string that has occurred in the plant $G$, and $\psi^{f_i}(\mu)$ tracks what the supervisor $i$ has observed thus far. 
%\begin{example}
%Let us take $\mu=\alpha\ tick\  f_{12}(\alpha) \beta g_{21}(1) \in \tilde{E}^*$.
%By the definitions of $\psi(\cdot)$ and $\psi^{f_2}(\cdot)$, we have $\psi(\mu)=\alpha\beta$ and $\psi^{f_2}(\mu)=tick\  \alpha\beta$.
%\end{example}
The following proposition says that $\tilde{G}$ does not change the system language of $G$.
\begin{proposition}\label{Prop1}
Given a networked DES $G$, we construct $\tilde{G}$ as described above. Then we have
$\psi(\mathcal{L}(\tilde{G}))=\mathcal{L}(G)$.
\end{proposition}
\begin{proof}
Please see Appendix A.
\end{proof}

The control objective $K$ is characterized by a specification automaton $H=(Q_H,\tilde{\Sigma},\delta_H,\Gamma_H,q_{0,H}, Q_{m,H}) \sqsubseteq G$ such that all the strings in $\mathcal{L}(G)$ are safe if they end in $Q_H$ and unsafe if they end in $Q/Q_H$.
%The distributed nonblocking networked supervisory control aims 
We define $\tilde{H}=(\tilde{Q}_H,\tilde{\Sigma},\tilde{\delta}_H,\tilde{\Gamma}_H,\tilde{q}_{0},\tilde{Q}_{m,H})$ as the accessible part of the automaton obtained from deleting all the states $({q},\bar{\theta})$ in $\tilde{G}$ with ${q} \in Q \setminus Q_H$.
To achieve $K=\mathcal{L}_m({H})$ under nondeterministic communication delays and losses, all the sequences  in $\mathcal{L}_m(\tilde{G})\setminus \mathcal{L}_m(\tilde{H})$ must be disabled  but sequences in  $\mathcal{L}_m(\tilde{H})$ cannot be disabled.
To achieve this, the supervisor $i$ is defined as a function $S_i: \psi^{f_i}(\mathcal{L}(\tilde{G})) \rightarrow 2^{\tilde{\Sigma}_{i}}$, where $S_i(t)$ is the set of events to be enabled given that the observed string for supervisor $i$ is $t$. We call $S_i(t)$ the control command that supervisor $i$ made when $t$ is observed.
Not all the supervisors are admissible because (i) we cannot disable an uncontrollable event, and (ii) we cannot disable $tick$ if
no enforceable event can pre-empted the occurrence of $tick$ at the current state, which yield the following definition.

\begin{definition}\label{Def3}
Given a set of supervisors $\gamma=[S_1,\ldots,S_n]$, $\gamma$ is said to be admissible if the following two conditions are satisfied for all $S_i$, $i\in \mathcal{A}$.
  \begin{enumerate}
    \item No uncontrollable events can be disabled, i.e., 
\begin{align}\label{Eq10}
(\forall \mu \in \mathcal{L}(\tilde{H}))\tilde{\Sigma}_{uc,i} \subseteq S_i(\psi^{f_i}(\mu));
\end{align}
    \item $tick$ is physically possible and no enforcement event can pre-empt it, then $tick$ cannot be disabled, i.e., 
\begin{align}\label{Eq11}
&(\forall \mu \in \mathcal{L}(\tilde{H}))[\tilde{\Gamma}_H(\tilde{\delta}_H(\tilde{q}_0,\mu)) \cap \Sigma_{for}=\emptyset]\wedge  \notag \\
&[tick \in \tilde{\Gamma}(\tilde{\delta}(\tilde{q}_0,\mu))] \Rightarrow  tick \in S_i(\psi^{f_i}(\mu)).
\end{align}
  \end{enumerate}
\end{definition}

%Not all sequences in  $\tilde{G}$ can occur under control.
%For any $\mu \in \mathcal{L}(\tilde{G})$ and any supervisor $i \in \mathcal{A}$, the control command taking effect at $G$ is $S_i(\psi^{f_i}(\mu))$.
%That is, for any event $\sigma \in \tilde{\Sigma}$ that is possible after $\mu \in \mathcal{L}(\tilde{G})$, it can occur after $\mu$ iff it is allowed by $S_i(\psi^{f_i}(\mu))$ for all $i\in \mathcal{A}^c(\sigma)$.
%Note that if $\sigma=tick$, $\mathcal{A}^c(tick)=\mathcal{A}$, and if $\sigma \in \tilde{\Sigma}\setminus\{tick\}$, $\mathcal{A}^c(\sigma)$ is unique.
For any $\mu \in \mathcal{L}(\tilde{G})$ and any supervisor $i \in \mathcal{A}$, the control command taking effect at $G$ is $S_i(\psi^{f_i}(\mu))$.
 Thus, under control of a set of admissible supervisors, all the sequences in  $\tilde{G}$ that can occur is given as follows.
 
\begin{definition}\label{Def2}
Given system $G$ and a set of admissible supervisors $\gamma=[S_1,\ldots,S_n]$, all the sequences in $\mathcal{L}(\tilde{G})$ that may occur under $\gamma$, denoted by $\mathcal{L}(\tilde{G},\gamma)$,   are defined as follows:
  \begin{enumerate}
    \item $\varepsilon \in \mathcal{L}(\tilde{G},\gamma)$;
    \item For any $\mu \in \mathcal{L}(\tilde{G},\gamma)$ and  $\sigma \in \tilde{E}$, $\mu \sigma \in \mathcal{L}(\tilde{G},\gamma)$ iff $\mu \sigma \in \mathcal{L}(\tilde{G})$ and  $[\sigma \in \tilde{\Sigma}_c \Rightarrow (\forall i\in \mathcal{A}^c(\sigma))\sigma \in S_i(\psi^{f_i}(\mu))]$.
  \end{enumerate}
  The language of $\mathcal{L}(\tilde{G})$ marked by $\gamma$ is defined as $\mathcal{L}_m(\tilde{G},\gamma)=\mathcal{L}(\tilde{G},\gamma)\cap\mathcal{L}_m(\tilde{G}).$
\end{definition}

%By Definition \ref{Def2}, $\mathcal{L}(\tilde{G},\gamma)$ models all the possible interactive behaviors of the distributed control system.

%By the definition of $\psi(\cdot)$, the closed-loop system behaviors can be simply decoded from  $\mathcal{L}(\tilde{G},\gamma)$ by removing all the events in $\Sigma^f \cup \Sigma^g$.
%That is, $\psi(\mathcal{L}(\tilde{G},\gamma))$ is the exact language that may be generated by the closed-loop system.

Note that if $\sigma=tick$, $\mathcal{A}^c(\sigma)=\mathcal{A}$; if $\sigma \in \tilde{\Sigma}_c\setminus \{tick\}$, $\mathcal{A}^c(\sigma)$ is unique.
The distributed nonblocking networked supervisory control problem (DNNSCP) is formulated as follows.
\begin{problem}\label{Prob1}
Given a DES $G$ with communication delays and losses between the supervisors, and a nonempty specification language $K \subseteq \mathcal{L}_m(G)$ modeled as sub-automaton $H \sqsubseteq G$.
Find a set of admissible nonblocking supervisors $\gamma=[S_1,\ldots,S_n]$ such that
$\mathcal{L}(\tilde{G},\gamma)=\mathcal{L}(\tilde{H})\ \text{and}\ \mathcal{L}_m(\tilde{G},\gamma)=\mathcal{L}_m(\tilde{H}).$
\end{problem}

\section{Distributed nonblocking supervisory control}

\subsection{Network controllability and network joint observability}

Next, we solve the DNNSCP.
First, to deal with uncontrollable events, we define the network controllability as follows.

%We show that for a set of supervisors to make correct decisions (i.e., to ensure that all the interactive behaviors of $\mathcal{L}_m(\tilde{H})$ can be achieved) based on partial observation, the following controllability and  networked joint observability conditions must be satisfied.

\begin{definition}\label{Def5}
Given networked DESs $H$ and $G$, we say that the specification language
$K$ is network controllable with respect to $\tilde{\Sigma}_{uc}$ and $G$, if 
1)  $\mathcal{L}(\tilde{H}) \tilde{\Sigma}_{uc}\cap \mathcal{L}(\tilde{G}) \subseteq \mathcal{L}(\tilde{H})$;
and 2)
\begin{align}\label{Eq12}
&(\forall \mu \in \mathcal{L}(\tilde{H}))\mu\ tick \in \mathcal{L}(\tilde{G})\setminus \mathcal{L}(\tilde{H})\Rightarrow  \notag \\
&\tilde{\Gamma}_{{H}}(\tilde{\delta}_{{H}}(\tilde{q}_0,\mu)) \cap \Sigma_{for}\neq \emptyset.
\end{align}
\end{definition}
Definition \ref{Def5} extends the controllability of timed DESs in \cite{lin95tac} to networked timed DESs.
Specifically, condition 1) says that all the uncontrollable events cannot be disabled.
Condition 2) says if we  disable $tick$ after a $\mu \in \mathcal{L}(\tilde{H})$,  there must exist some enforceable events that are available after $\mu$ in $\tilde{H}$.

Second,  the nondetermined observations of the supervisors impose further limitations on behaviors that can be achieved by the partially observed supervisors.
Thus, in addition to network controllability, an additional condition on $\mathcal{L}(\tilde{H})$ and $\mathcal{L}(\tilde{G})$ called network joint observability is proposed as follows.

\begin{definition}
Given networked DESs $H$ and $G$,  we say that the specification language
$K$  is network jointly observable with respect to $G$ and $\psi^{f_i}(\cdot)$, $i \in \mathcal{A}$, if
\begin{align}\label{Eq13}
&(\forall \mu \in \mathcal{L}(\tilde{H}))(\forall \sigma \in \tilde{\Sigma}_{c})\mu\sigma \in \mathcal{L}(\tilde{G}) \setminus \mathcal{L}(\tilde{H})  \notag \\
&\Rightarrow  [(\forall i \in \mathcal{A}^c(\sigma))(\forall \nu \sigma\in \mathcal{L}(\tilde{H}))\psi^{f_i}(\mu) \neq \psi^{f_i}(\nu)].
\end{align}
%and
%\begin{align}\label{Eq13}
%&[(\forall s,s' \in \mathcal{L}(\tilde{H}))s\ tick \in \mathcal{L}(\tilde{H}) \wedge s'\ tick \in \mathcal{L}(\tilde{G})\setminus \mathcal{L}(\tilde{H})] \notag \\
%&\Rightarrow [(\forall i \in \mathcal{A})\psi^{f_i}(s) \neq \psi^{f_i}(s')].
%\end{align}
\end{definition}
Equation (\ref{Eq13}) says that if a controllable event $\sigma$ (including $tick$) needs to be disabled after a $\mu \in \mathcal{L}(\tilde{H})$,
all the supervisors who can disable $\sigma$ must distinguish $\mu$ from all the $\nu \in \mathcal{L}(\tilde{H})$ after which $\sigma$ needs to be enabled.

%Since $G_1,\ldots, G_n$ are independent (except for $tick$), if $\sigma \in \tilde{\Sigma}_c \setminus \{tick\}$, there exists one and only one supervisor who can control the occurrence of $\sigma$, i.e., $|\mathcal{A}^c(\sigma)|=1$. 
%By (\ref{Eq13}), it requires the supervisor $\mathcal{A}^c(\sigma)$ to distinguish $\mu$ and $\mu'$.
%Otherwise, since all the $G_1,\ldots, G_n$ share a common $tick$, 
%Note that if $\sigma=tick$,  $\mathcal{A}^c(tick)=\mathcal{A}$, and if $\sigma \in \tilde{\Sigma}\setminus\{tick\}$, $\mathcal{A}^c(\sigma)$ is unique.
%By (\ref{Eq13}),  all the supervisors should distinguish $\mu$ and $\mu'$.

%Meanwhile, since all the $G_i$ share a common $tick$, i.e., $(\forall i\in \mathcal{A})tick \in \tilde{\Sigma}_{c,i}$, (\ref{Eq13}) requires that all the supervisors must distinguish two strings $s, s' \in \mathcal{L}(\tilde{H})$ such that the enablement and disablement decisions on $e=tick$ after $s,s'$ are different.
%\begin{definition}
%The set of nonblocking distributed supervisors is denoted as $\mathcal{S}^*=[S^*_1,S^*_2,\ldots,S^*_n]$, and for each $i \in \mathcal{A}$, $S^*_i$ is defined as follows. For any $t_i \in \psi^{f_i}(\mathcal{L}(\tilde{G}))$, we define
%\begin{align}\label{Eq11}
%S^*_i(t_i)=\tilde{\Sigma}_{uc,i} \cup \{\sigma \in \tilde{\Sigma}_{c,i}: (\exists\mu\sigma \in \mathcal{L}(\tilde{H}))[\psi^{f_i}(\mu)=t_i]\}.
%\end{align}
%\end{definition}

\begin{theorem}\label{Theo1}
Consider a networked DES $G$.
%Let $\psi^{f_i}(\cdot)$ be the networked projection from $\tilde{\Sigma}^*$ to $\Sigma_o$, $i=1,\ldots,n$.
For a nonempty $K \subseteq \mathcal{L}_m(G)$ modeled as a sub-automaton $H \sqsubseteq G$, DNNSCP is solvable
if and only if the  following three conditions hold:
  \begin{enumerate}
    \item  $K$ is network controllable with respect to $G$ and $\tilde{\Sigma}_{uc}$;
    \item $K$ is network jointly observable with respect to $G$ and $\psi^{f_i}(\cdot)$, $i\in \mathcal{A}$;
\item $\mathcal{L}(\tilde{H})$ is $\mathcal{L}_m(\tilde{G})$-closed, i.e., $\mathcal{L}_m(\tilde{H})=\mathcal{L}(\tilde{H})\cap \mathcal{L}_m(\tilde{G})$.
  \end{enumerate}
  Furthermore, if DNNSCP is solvable, the solution is given as follows:
  For all $t \in \psi^{f_i}(\mathcal{L}(\tilde{G}))$, 
\begin{align}\label{Eq14}
S_i(t)=&\tilde{\Sigma}_{uc,i} \cup (\tilde{\Sigma}_{c,i}\setminus\{\sigma \in \tilde{\Sigma}_{c,i}: (\exists \mu \in \mathcal{L}(\tilde{H})) \notag \\
&\mu\sigma \in \mathcal{L}(\tilde{G})\setminus \mathcal{L}(\tilde{H}) \wedge \psi^{f_i}(\mu)=t\}).
\end{align}
\end{theorem}
\begin{proof}
Please see Appendix B.
\end{proof}

\subsection{Application}

\begin{figure}
	\begin{center}
		\includegraphics[width=5.0cm]{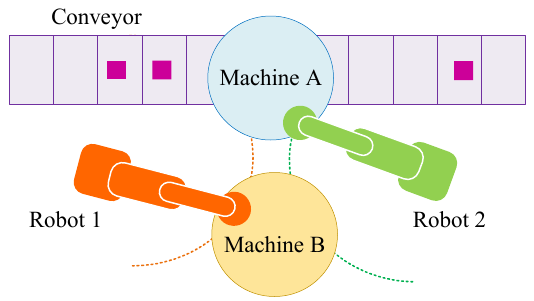}    % The printed column width is 8.4 cm.
		\caption{A production line with two robots.} 
		\label{Fig6}
	\end{center}
\end{figure}

We consider the production line depicted in Fig. \ref{Fig6}, which consists of machines A and B, and a conveyor.
Each part needs to be successively processed by machines A, B, and A. That is, a part will be done in three steps, accomplished by  machines A, B, and A, respectively.
Two robots, named robots 1 and 2, cooperatively work at the production line.
Specifically, robot 1 first takes a part from the conveyor and puts it on machine A, and machine A  starts to process this part.
When machine A finishes the first step, robot 1 moves this part  to machine B, and machine B continues to process it. 
When machine B finishes the second step, robot 2 moves this part back to machine A for accomplishing the third step.
The part is done if machine A finishes the third step,  and then robot 2  places it back to the conveyor.
The above process is repeated again and again.
Each machine can only process one part at a time.
If  machine B is processing a part and a new part is sent to machine A, the system will get stuck as neither of these two parts can be further processed, which should be prevented from happening.

\begin{figure}
\centering 
\subfigure[Systems $G_i,i=1,2$]{\label{Fig41}\includegraphics[width=2.9cm]{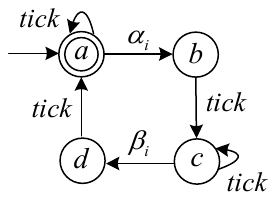}}
\subfigure[System $G$]{\label{Fig42}\includegraphics[width=5.6cm]{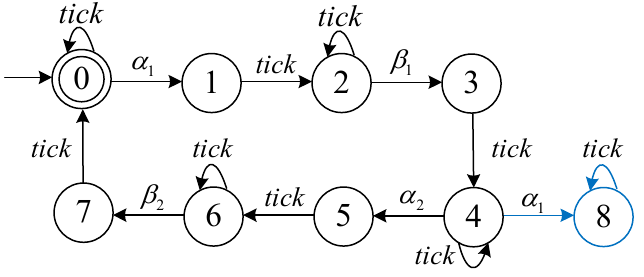}}
\subfigure[The communication automaton $\tilde{G}$]{\label{Fig43}\includegraphics[width=8.8cm]{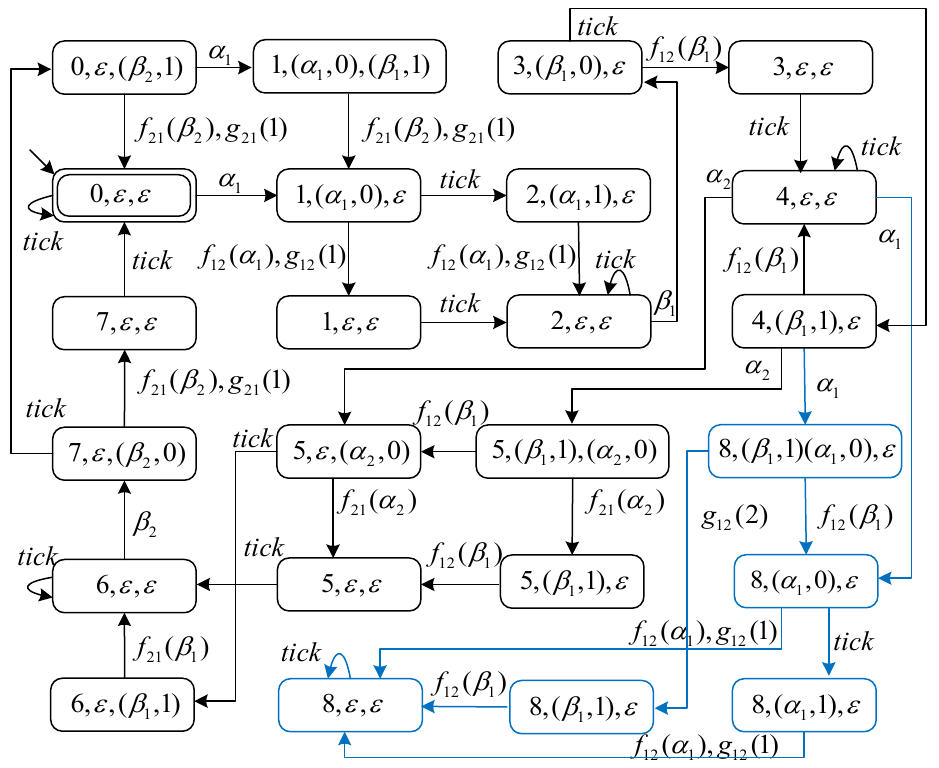}}
\caption{Automata $G_i$, $i=1,2$, $G$, and $\tilde{G}$.} \label{Fig4}
\end{figure}

System $G_i$ for robot $i$ is constructed in Fig. \ref{Fig41}.
%We interpret the construction of $G$ as follows.
The event set for subsystem $G_i$ is $\tilde{\Sigma}_i=\{\alpha_i, \beta_i, tick\}$, $i=1,2$ where $\alpha_1$: robot 1 takes a part from the conveyor and places it on machine A; $\beta_1$: robot 1 takes the part away from machine A and puts it on machine B; $\alpha_2$: robot 2 takes the part away from machine B and places it on machine A; $\beta_2$: robot 2 puts this part back to the conveyor from machine A. 
We assume that it takes at least one unit of time to finish  $\alpha_1$, $\alpha_2$, $\beta_1$, and $\beta_2$.
As we can see  in  $G_i$, $\beta_i$ can occur after $\alpha_i$ only if one $tick$ has occurred, $i=1,2$.
The overall system can be obtained as $G_1||G_2$.
Some strings in $G_1||G_2$ are not feasible as a robot can take away a part from a machine only if there is a part on the machine.
We refine the structure of $G_1||G_2$ by removing all the physically unfeasible strings and obtain   $G$, as depicted in Fig. \ref{Fig42}.
In state 8 of $G$, the system gets stuck as mentioned above.
The desired system $H$ can be obtain by removing all the transitions of $G$ connected to state 8 (highlighted by blue in Fig. \ref{Fig42}).

We assume that robots 1 and 2 are controlled by supervisors 1 and  2, respectively.
Supervisors 1 and 2 can communicate with each other, i.e., $\mathbf{COM}_{12}=\mathbf{COM}_{21}=1$.
We define $\tilde{\Sigma}_{o,i}=\tilde{\Sigma}_{c,i}=\{\alpha_i,\beta_i, tick\}$, $i=1,2$.
We also define  $\Sigma_{12}=\{\alpha_1,\beta_1\}$ and $\Sigma_{21}=\{\alpha_2,\beta_2\}$.
The communication delays for both $\mathbf{CH}_{12}$ and $\mathbf{CH}_{21}$ are  upper bounded by $1$, i.e., $N_{12}=N_{21}=1$.
We define $\Sigma_{L,12}=\{\alpha_1\}$ and $\Sigma_{L,21}=\{\beta_2\}$ as the set of events that may be lost.
The communication automaton  $\tilde{G}$ is constructed in Fig. \ref{Fig43}.
$\tilde{H}$ can be obtained by deleting all the illegal transitions of $\tilde{G}$ (highlighted by blue).
By Fig. \ref{Fig43}, we must disable $\alpha_1$  at states $(4,\varepsilon,\varepsilon)$ and $(4,(\beta_1,1),\varepsilon)$ but enable $\alpha_1$ at states $(0,\varepsilon,\varepsilon)$ and $(0,\varepsilon,(\beta_2,1))$.
Since the occurrence of $\alpha_2$ must be communicated to supervisor $1$ (modeled by $f_{21}(\alpha_2)$) before $\alpha_1$ occurs at $(0,\varepsilon,\varepsilon)$ and $(0,\varepsilon,(\beta_2,1))$, supervisor 1 can always distinguish strings ending up at state $4$ and strings ending up at state $0$.
Therefore, $K=\mathcal{L}_m(H)$ is network jointly observable with respect to $G$.
Since all the events are controllable, the first condition of network controllability is trivially true.
Moreover, since $\mu \in \mathcal{L}(\tilde{H})\wedge \mu\ tick \in \mathcal{L}(\tilde{G})\Rightarrow \mu\ tick \in \mathcal{L}(\tilde{H})$, $K=\mathcal{L}_m(H)$ is network controllable with respect to $G$. 
Additionally,  $\mathcal{L}(\tilde{H})$ is $\mathcal{L}_m(\tilde{G})$-closed.
By Theorem \ref{Theo1}, a set of admissible nonblocking supervisors $\gamma=[S_1,S_2]$ can be obtained as in (\ref{Eq14}).

\section{Conclusion}
%In distributed  control, communication delays and losses are often unavoidable when
%the communication between the supervisors is carried out over some shared networks. 
In this paper, we have considered the distributed supervisory control problem with nondeterministic delays and losses.
The system has been modeled as a timed automaton, where ``$tick$'' is used to characterize the elapse of one unit of time.
Under the assumption that the communication delays are upper bounded by a finite number of $tick$ occurrences, we have developed procedures to define observations of the supervisors by taking the communication delays and losses into consideration.
Based on the developed procedures, we have derived the necessary and sufficient conditions for
the existence of a set of admissible supervisors to achieve the control objective deterministically.
Finally,  a practical example has also been provided to show the application of the proposed framework.

\bibliographystyle{IEEEtran}     
\bibliography{BibFiles/articles}

\appendix

\subsection{Proof of Proposition \ref{Prop1}}
\begin{proof}
We first prove $\psi(\mathcal{L}(\tilde{G})) \subseteq \mathcal{L}(G)$.
For any $\mu \in \mathcal{L}(\tilde{G})$, we write $\tilde{\delta}(\tilde{q}_0,\mu)=({q},\bar{\theta})$ for ${q}\in {Q}$ and $\bar{\theta}\in\Theta$.
We now prove $\delta(q_0,\psi(\mu))={q}$ by induction on the finite length of sequences in $\mathcal{L}(\tilde{G})$.

Since $\tilde{\delta}(\tilde{q}_0,\varepsilon)=\tilde{q}_0=(q_0, \bar{\theta}_0)$ and $\psi(\varepsilon)=\varepsilon$ and $\delta(q_0,\varepsilon)=q_0$, the base case is true.

The induction hypothesis is that for any $\mu \in \mathcal{L}(\tilde{G})$ with $|\mu| \le k$, if $\tilde{\delta}(\tilde{q}_0,\mu)=(q,\bar{\theta})$, then $\delta(q_0,\psi(\mu))=q$.
We next prove the same is also true for $\mu e \in \mathcal{L}(\tilde{G})$ with $|\mu|=k$.
We write $\tilde{\delta}(\tilde{q}_0,\mu e)=(q',\bar{\theta}')$.
By definition, (i) $e \in \tilde{\Sigma}$ or (ii) $e\in \Sigma^f \cup \Sigma^g$.
We consider each of them separately as follows.

If $e \in \tilde{\Sigma}$, by (\ref{Eq5}) and (\ref{Eq6}), we have $q'=\delta(q,e)$.
By the induction hypothesis, $q=\delta(q_0,\psi(\mu))$.
By the definition of $\psi(\cdot)$, $\psi(\mu e)=\psi(\mu)e$.
Thus, $\delta(q_0,\psi(\mu e))=\delta(\delta(q_0,\psi(\mu)),e)=\delta(q,e)=q'$.

Otherwise, if $e\in \Sigma^f \cup \Sigma^g$, by (\ref{Eq7}) and (\ref{Eq8}), $q'=q$.
By the induction hypothesis, $q=\delta(q_0,\psi(\mu))$.
By the definition of $\psi(\cdot)$, $\psi(\mu e)=\psi(\mu)$.
Thus, $\delta(q_0,\psi(\mu e))=\delta(q_0,\psi(\mu))=q=q'$.

We next prove $\psi(\mathcal{L}(\tilde{G})) \supseteq \mathcal{L}(G)$.
Let us denote $E=\Sigma_{11}\cup\cdots\Sigma_{1n} \cup \cdots \cup \Sigma_{n1} \cup \cdots \cup \Sigma_{nn}$ by the set of all the communication events.
For any $s \in \mathcal{L}(G)$, we write $s=t_0\sigma_1t_1 \cdots t_{k-1}\sigma_k t_k$ for $t_z \in (\tilde{\Sigma} \setminus E)^*$ and $\sigma_z \in E$.
Since $\sigma_z \in E$, without loss of generality, we write $\sigma_z \in \Sigma_{i_z j_z}$ for $z=1,\ldots, k$.
By the definition of $\tilde{G}$, one can check that $\mu=t_0\sigma_1f_{i_1j_1}(\sigma_1)t_1 \cdots t_{k-1}\sigma_k f_{i_kj_k}(\sigma_k) t_k \in \mathcal{L}(\tilde{G})$.
By the definition of $\psi(\cdot)$, we have $\psi(\mu)=s$.
Since $s$ is arbitrarily given, $\psi(\mathcal{L}(\tilde{G})) \supseteq \mathcal{L}(G)$.
\end{proof}

\subsection{Proof of Theorem \ref{Theo1}}
\begin{proof}

We first prove that $\gamma=[S_1,\ldots,S_n]$ is admissible.
By Definition \ref{Def3}, we need to prove that $S_i$ satisfies (\ref{Eq10}) and (\ref{Eq11})  for all $i\in \mathcal{A}$.
Since $(\forall t \in \psi^{f_i}(\mathcal{L}(\tilde{G})))\tilde{\Sigma}_{uc,i} \subseteq S_i(t)$ and $\mathcal{L}(\tilde{G},\gamma)\subseteq \mathcal{L}(\tilde{G})$,  $(\forall \mu \in \mathcal{L}(\tilde{G},\gamma))\tilde{\Sigma}_{uc,i}\subseteq S_i(\psi^{f_i}(\mu))$, which implies that condition (\ref{Eq10}) of Definition \ref{Def3} is true.
We now prove that $S_i$ satisfies condition (\ref{Eq11}).
The proof is by contradiction.
Suppose that (\ref{Eq11}) is not true, i.e., 
\begin{align*}
&(\exists \mu \in \mathcal{L}(\tilde{H}))[\tilde{\Gamma}_H(\tilde{\delta}_H(\tilde{q}_0,\mu)) \cap \Sigma_{for}=\emptyset]\wedge \notag \\
& [tick \in \tilde{\Gamma}(\tilde{\delta}(\tilde{q}_0,\mu))] \wedge  [tick \notin S_i(\psi^{f_i}(\mu))].
\end{align*}
%Since $\mu \in \mathcal{L}(\tilde{G},\gamma)$ and $\mathcal{L}(\tilde{G},\gamma)=\mathcal{L}(\tilde{H})$, we have $\mu \in \mathcal{L}(\tilde{H})$.
Since $tick \notin S_i(\psi^{f_i}(\mu))$, by (\ref{Eq14}), there exists $\nu \in \mathcal{L}(\tilde{H})$ such that $\nu \ tick \in \mathcal{L}(\tilde{G})\setminus \mathcal{L}(\tilde{H})$ and $\psi^{f_i}(\mu)=\psi^{f_i}(\nu)$.
Since $tick \in \tilde{\Gamma}(\tilde{\delta}(\tilde{q}_0,\mu))$, we have $\mu\ tick \in \mathcal{L}(\tilde{G})$.
Furthermore, by network joint observability, $\mu\ tick \notin \mathcal{L}(\tilde{H})$, because otherwise,
\begin{align*}
&(\exists \mu, \nu \in \mathcal{L}(\tilde{H}))(\exists \sigma=tick \in \tilde{\Sigma}_{c}) \mu\sigma\in \mathcal{L}(\tilde{H}) \wedge\notag \\
&\nu\sigma \in \mathcal{L}(\tilde{G}) \setminus \mathcal{L}(\tilde{H})  \wedge [(\exists i \in \mathcal{A}^c(\sigma))\psi^{f_i}(\nu)=\psi^{f_i}(\mu)],
\end{align*}
which contradicts that $K$ is network jointly observable with respect to $\mathcal{L}({G})$.
Overall, there exists $\mu \in \mathcal{L}(\tilde{H})$ such that $\mu \ tick \in \mathcal{L}(\tilde{G}) \setminus \mathcal{L}(\tilde{H})$ and $\tilde{\Gamma}_H(\tilde{\delta}_H(\tilde{q}_0,\mu)) \cap \Sigma_{for}=\emptyset$, which contradicts that $K$ is network controllable with respect to $\mathcal{L}({G})$.
Thus,  $S_i$ satisfies (\ref{Eq10}) and (\ref{Eq11}) for all $i\in \mathcal{A}$, which implies that $\gamma=[S_1,\ldots,S_n]$ is admissible.

(\textbf{IF Part}) 
Let us first prove $\mathcal{L}(\tilde{G},\gamma)=\mathcal{L}(\tilde{H})$ by proving that $\mathcal{L}(\tilde{G},\gamma) \subseteq \mathcal{L}(\tilde{H})$ and $\mathcal{L}(\tilde{G},\gamma) \supseteq \mathcal{L}(\tilde{H})$.
The proof is by induction on the length of the strings in $\mathcal{L}(\tilde{G},\gamma)$ and $\mathcal{L}(\tilde{H})$.

($\supseteq$) Since $\varepsilon \in \mathcal{L}(\tilde{H})$ and $\varepsilon \in \mathcal{L}(\tilde{G},\gamma)$, $\varepsilon \in \mathcal{L}(\tilde{H}) \Rightarrow \varepsilon \in \mathcal{L}(\tilde{G},\gamma)$.
The induction hypothesis is that for any $\mu \in \mathcal{L}(\tilde{H})$ such that $|\mu| \le l$, $\mu \in \mathcal{L}(\tilde{H}) \Rightarrow \mu \in \mathcal{L}(\tilde{G},\gamma)$.
We next prove that the same is also true for $\mu e \in \mathcal{L}(\tilde{H})$ such that $|\mu e|=l+1$ as follows.
Since $e \in \tilde{E}$, we have $e \in \Sigma^f \cup \Sigma^g \cup \tilde{\Sigma}_{uc}$ or $e \in \tilde{\Sigma}_c$.
We consider each of them separately as follows.

\textsl{Case 1:} $e \in \Sigma^f \cup \Sigma^g \cup \tilde{\Sigma}_{uc}$. 
Since $\mu e \in \mathcal{L}(\tilde{H})$ and $\mathcal{L}(\tilde{H}) \subseteq \mathcal{L}(\tilde{G})$, we have $\mu e \in \mathcal{L}(\tilde{G})$.
Since $\mu \in \mathcal{L}(\tilde{G},\gamma)$ and  $\mu e \in \mathcal{L}(\tilde{G})$ and $e \in \Sigma^f \cup \Sigma^g \cup \tilde{\Sigma}_{uc}$, by Definition \ref{Def2}, $\mu e \in \mathcal{L}(\tilde{G},\gamma)$.

\textsl{Case 2:} $e \in \tilde{\Sigma}_c$. 
Since $\mu e \in \mathcal{L}(\tilde{H})$, by network joint observability,  $(\forall i \in \mathcal{A}^c(e))(\forall \nu \in \mathcal{L}(\tilde{H}))\nu e \in \mathcal{L}(\tilde{G})\setminus \mathcal{L}(\tilde{H})\Rightarrow \psi^{f_i}(\mu)\neq \psi^{f_i}(\nu)$.
Thus, by (\ref{Eq14}),  $e \in S_i(\psi^{f_i}(\mu))$ for all $i \in \mathcal{A}^c(e)$.
Moreover, since  $\mu \in \mathcal{L}(\tilde{G},\gamma)$ and  $\mu e \in \mathcal{L}(\tilde{G})$, by Definition \ref{Def2}, $\mu e \in \mathcal{L}(\tilde{G},\gamma)$.
Therefore, $\mu e \in \mathcal{L}(\tilde{H}) \Rightarrow \mu e \in \mathcal{L}(\tilde{G},\gamma)$.

($\subseteq$) Since $\varepsilon \in \mathcal{L}(\tilde{G},\gamma)$ and $\varepsilon \in \mathcal{L}(\tilde{H})$, $\varepsilon \in \mathcal{L}(\tilde{G},\gamma) \Rightarrow \varepsilon \in \mathcal{L}(\tilde{H})$.
The induction hypothesis is that for any $\mu \in \mathcal{L}(\tilde{G},\gamma)$ such that $|\mu| \le l$,  $\mu \in  \mathcal{L}(\tilde{G},\gamma) \Rightarrow \mu \in \mathcal{L}(\tilde{H})$.
We next prove the same is also true for $\mu e \in \mathcal{L}(\tilde{G},\gamma)$ such that $|\mu e|=l+1$.
By definition, $e\in \Sigma^f \cup \Sigma^g$ or $e \in \tilde{\Sigma}$.
We consider each of them separately as follows.

\textsl{Case 1:} $e\in \Sigma^f \cup \Sigma^g$. 
Since $\mu, \mu e \in \mathcal{L}(\tilde{G})$, without loss of generality, we write $\tilde{\delta}(\tilde{q}_0,\mu)=({q},\bar{\theta})$ and $\tilde{\delta}(\tilde{q}_0,\mu)=({q}',\bar{\theta}')$ for $q,q'\in Q$ and $\bar{\theta},\bar{\theta}'\in \Theta$.
By the definition of $\psi(\cdot)$,  $\psi(\mu e)=\psi(\mu)$.
By Proposition 1, ${q}=\delta(q_0,\psi(\mu))=\delta(q_0,\psi(\mu e))={q}'$.
Since $\mu \in \mathcal{L}(\tilde{H})$, ${q}={q}'\in Q_H$.
Thus, $\mu e\in \mathcal{L}(\tilde{H})$.

\textsl{Case 2:} $e \in \tilde{\Sigma}$.
If $e \in \tilde{\Sigma}_{uc}$, since $\mu \in \mathcal{L}(\tilde{H})$ and $\mu e \in \mathcal{L}(\tilde{G},\gamma)\subseteq\mathcal{L}(\tilde{G})$, the network controllability immediately yields $\mu e \in \mathcal{L}(\tilde{H})$.
If $e \in \tilde{\Sigma}_{c}$, since  $\mu e \in \mathcal{L}(\tilde{G},\gamma)$, by Definition \ref{Def2},   $e \in S_i(\psi^{f_i}(\mu))$ for all $i\in \mathcal{A}^c(e)$.
By (\ref{Eq14}), for all $\nu \in \mathcal{L}(\tilde{H})$ and $\nu e \in \mathcal{L}(\tilde{G})\setminus \mathcal{L}(\tilde{H})$, we have $\psi^{f_i}(\mu) \neq \psi^{f_i}(\nu)$, $i \in \mathcal{A}^c(e)$.
Therefore,  $\mu e \in \mathcal{L}(\tilde{H})$.

Overall, we have $\mathcal{L}(\tilde{G},\gamma) \subseteq \mathcal{L}(\tilde{H})$ and $\mathcal{L}(\tilde{G},\gamma) \supseteq \mathcal{L}(\tilde{H})$, which implies $\mathcal{L}(\tilde{G},\gamma)=\mathcal{L}(\tilde{H})$.

We next prove $\mathcal{L}_m(\tilde{G},\gamma)=\mathcal{L}_m(\tilde{H})$.
By definition, $\mathcal{L}_m(\tilde{G},\gamma)=\mathcal{L}(\tilde{G},\gamma)\cap \mathcal{L}_m(G)$.
Moreover, since $\mathcal{L}(\tilde{G},\gamma)=\mathcal{L}(\tilde{H})$, we have $\mathcal{L}_m(\tilde{G},\gamma)=\mathcal{L}(\tilde{H}) \cap \mathcal{L}_m(G)$.
Since $\mathcal{L}(\tilde{H})$ is $\mathcal{L}_m(\tilde{G})$-closed, $\mathcal{L}_m(\tilde{H})=\mathcal{L}(\tilde{H})\cap \mathcal{L}_m(\tilde{G})$.
Thus, $\mathcal{L}_m(\tilde{G},\gamma)=\mathcal{L}_m(\tilde{H})$.

%We let $\gamma=[S_1,\ldots,S_n]$.

(\textbf{ONLY IF Part}) We assume that $\gamma=[S_1,\ldots, S_n]$ given
by: for $t \in \psi^{f_i}(\mathcal{L}(\tilde{G}))$, 
\begin{align*}
S_i(t)=&\tilde{\Sigma}_{uc,i} \cup (\tilde{\Sigma}_{c,i}\setminus\{\sigma \in \tilde{\Sigma}_{c,i}: (\exists \mu \in \mathcal{L}(\tilde{H})) \notag \\
&\mu\sigma \in \mathcal{L}(\tilde{G})\setminus \mathcal{L}(\tilde{H}) \wedge \psi^{f_i}(\mu)=t\}),
\end{align*}
solves the DNNSCP, i.e., $\gamma$ is admissible, $\mathcal{L}(\tilde{G},\gamma)=\mathcal{L}(\tilde{H})$, and $\mathcal{L}_m(\tilde{G},\gamma)=\mathcal{L}_m(\tilde{H}).$
We now prove 1), 2), and 3) of Theorem \ref{Theo1}.

{Network joint observability}:
By contradiction. we assume that (i) there exist $\mu \in \mathcal{L}(\tilde{H})$ and $\sigma \in \tilde{\Sigma}_c$ such that $\mu \in \mathcal{L}(\tilde{H})$ and $\mu \sigma \in \mathcal{L}(\tilde{G}) \setminus \mathcal{L}(\tilde{H})$; (ii) there exist $i \in \mathcal{A}^c(\sigma)$, $\mu_i\sigma \in \mathcal{L}(\tilde{H})$  such that $\psi^{f_i}(\mu)=\psi^{f_i}(\mu_i)$.
By (\ref{Eq14}), $\sigma \notin S_i(\psi^{f_i}(\mu))=S_i(\psi^{f_i}(\mu_i))$.
Moreover, since $\mu_i \in \mathcal{L}(\tilde{H})=\mathcal{L}(\tilde{G},\gamma)$ and $\mu_i \sigma \in \mathcal{L}(\tilde{H})\subseteq \mathcal{L}(\tilde{G})$, by Definition \ref{Def2},  $\mu_i \sigma \in \mathcal{L}(\tilde{G})\setminus \mathcal{L}(\tilde{G},\gamma)$.
Since $\mu_i \sigma \in \mathcal{L}(\tilde{H})$, we have $\mathcal{L}(\tilde{G},\gamma)\neq \mathcal{L}(\tilde{H})$, which contradicts $\mathcal{L}(\tilde{G},\gamma)=\mathcal{L}(\tilde{H})$.

{Network controllability}: 
We first prove that the first condition of network controllability is true, i.e., $\mathcal{L}(\tilde{H})\tilde{\Sigma}_{uc}\cap \mathcal{L}(\tilde{G}) \subseteq \mathcal{L}(\tilde{H})$.
To do this, we take arbitrary $\mu \in \mathcal{L}(\tilde{H})=\mathcal{L}(\tilde{G},\gamma)$ and $\sigma \in \tilde{\Sigma}_{uc}$ such that $\mu\sigma\in \mathcal{L}(\tilde{G})$.
Since $\sigma \in \tilde{\Sigma}_{uc}$, we have $\sigma \in \tilde{\Sigma}_{uc,i}$ for some $i \in \mathcal{A}$.
By (\ref{Eq14}), $\sigma \in S_i(\psi^{f_i}(\mu))$.
Since $\mu \in \mathcal{L}(\tilde{G},\gamma)$, $\mu \sigma \in \mathcal{L}(\tilde{G})$, and $\sigma \in S_i(\psi^{f_i}(\mu))$, by Definition \ref{Def2}, $\mu\sigma \in \mathcal{L}(\tilde{G},\gamma)= \mathcal{L}(\tilde{H})$.
Since  $\mu \in \mathcal{L}(\tilde{H})$ and $\sigma \in \tilde{\Sigma}_{uc}$ are arbitrarily given,  $\mathcal{L}(\tilde{H}) \tilde{\Sigma}_{uc}\cap \mathcal{L}(\tilde{G}) \subseteq \mathcal{L}(\tilde{H})$.

We now prove that the second condition of network controllability is also true, i.e., (\ref{Eq12}) holds.
The proof is by contradiction.
Suppose that (\ref{Eq12}) is not true, i.e., 
there exists $\mu \in \mathcal{L}(\tilde{H})$ such that $\mu \ tick \in \mathcal{L}(\tilde{G}) \setminus \mathcal{L}(\tilde{H})$ and $\tilde{\Gamma}_{{H}}(\tilde{\delta}_{{H}}(\tilde{q}_0,\mu)) \cap \Sigma_{for}=\emptyset.$
Since $\mu \in \mathcal{L}(\tilde{H})$ and $\mu \ tick \in \mathcal{L}(\tilde{G}) \setminus \mathcal{L}(\tilde{H})$, by (\ref{Eq14}), we have $tick \notin S_i(\psi^{f_i}(\mu))$ for all $i\in \mathcal{A}^c(tick)=\mathcal{A}$.
Moreover, since $\mu \ tick \in \mathcal{L}(\tilde{G})$, we have $tick \in \tilde{\Gamma}(\tilde{\delta}(\tilde{q}_0,\mu))$.
Thus, there exists $\mu \in \mathcal{L}(\tilde{H})$ such that $\tilde{\Gamma}_{{H}}(\tilde{\delta}_{{H}}(\tilde{q}_0,\mu)) \cap \Sigma_{for}=\emptyset$ and $tick \in \tilde{\Gamma}(\tilde{\delta}(\tilde{q}_0,\mu))$ and $tick \notin S_i(\psi^{f_i}(\mu))$ for all $i\in \mathcal{A}$. By Definition \ref{Def3},  $\gamma=[S_1,\ldots, S_n]$ is not admissible, which is a contradiction.
Therefore, the second condition of the network controllability is also true.
By Definition \ref{Def5}, $\mathcal{L}(\tilde{H})$ is network controllable with respect to $\mathcal{L}(\tilde{G})$ and $\tilde{\Sigma}_{uc}$.

{$\mathcal{L}_m(\tilde{G})$-closure}:
 By the definition of $\mathcal{L}_m(\tilde{G},\gamma)$, $\mathcal{L}_m(\tilde{G},\gamma)=\mathcal{L}(\tilde{G},\gamma) \cap \mathcal{L}_m(G)$. 
 Since $\mathcal{L}_m(\tilde{G},\gamma)=\mathcal{L}_m(\tilde{H})$ and  $\mathcal{L}(\tilde{G},\gamma)=\mathcal{L}(\tilde{H})$, it has $\mathcal{L}_m(\tilde{H})=\mathcal{L}(\tilde{H}) \cap \mathcal{L}_m(\tilde{G})$, which is the $\mathcal{L}_m(\tilde{G})$-closure condition.
\end{proof}

\end{document}